\documentclass[11pt, oneside ]{article}
\usepackage{STpaperformat}
\usepackage{STpapersymbol}
\usepackage{color}
\usepackage{times}
\usepackage[margin=1in]{geometry}                		
\usepackage{graphicx}				
\usepackage{amsmath,amssymb} 
\usepackage{bm}  
\usepackage{txfonts}
\usepackage[T1]{fontenc}
\usepackage{mathbbol}

\usepackage[
bookmarks=true,
bookmarksnumbered=true,
bookmarkstype=toc,
colorlinks,
breaklinks,
pdfauthor=Seiichiro Tani,%
]{hyperref}
\definecolor{dullmagenta}{rgb}{0.4,0,0.4}   
\definecolor{darkblue}{rgb}{0,0,0.4}
\hypersetup{linkcolor=darkblue,citecolor=blue,filecolor=dullmagenta,urlcolor=darkblue} 
\def\equationautorefname~#1\null{\textrm{Eq.~(#1)}\null}
\def\figureautorefname~#1\null{\textrm{Fig.~#1}\null}
\def\tableautorefname~#1\null{\textrm{Tab.~#1}\null}
\def\sectionautorefname~#1\null{\textrm{Sec.~#1\;}\null}
\def\subsectionautorefname~#1\null{\textrm{Sec.~#1\;}\null}
\def\subsubsectionautorefname~#1\null{\textrm{Sec.~#1\;}\null}
\def\pageautorefname~#1\null{\textrm{page~#1\;} \null}

\newcommand{\INCREMENT}{\mathrm{INC}}
\newcommand{\Identity}{\Id}
\newcommand{\DQC}{\mathrm{DQC}}

\author{Seiichiro Tani\\
NTT Communication Science Laboratories, NTT Corporation, Japan \and International Research Frontiers Initiative (IRFI),  Tokyo Institute of Technology, Japan\\
\href{mailto:seiichiro.tani.cs@hco.ntt.co.jp}{seiichiro.tani.cs@hco.ntt.co.jp} 
}
\date{}
\title{Space-Bounded Unitary Quantum Computation with Postselection}

\begin{document}
\pagestyle{plain}
\sloppy
\maketitle
\begin{abstract}
Space-bounded computation has been a central topic in classical and quantum complexity theory.
In the quantum case, every elementary gate must be unitary. This restriction makes it unclear whether the power of  space-bounded computation changes by allowing intermediate measurement. In the bounded error case, Fefferman and Remscrim [STOC 2021, pp.1343--1356] 
and Girish, Raz and Zhan~[ICALP 2021, pp.73:1--73:20] recently provided the break-through results that the power does not change. This paper shows that a similar result holds for space-bounded quantum computation with \emph{postselection}. Namely,  it is proved possible to eliminate intermediate postselections and measurements in the space-bounded quantum computation in the bounded-error setting. Our result strengthens the recent result by 
Le Gall, Nishimura and Yakaryilmaz~[TQC 2021, pp.10:1--10:17] that logarithmic-space bounded-error quantum computation with \emph{intermediate} postselections and measurements is equivalent in computational power to logarithmic-space unbounded-error probabilistic computation. As an application, it is shown that  bounded-error space-bounded one-clean qubit computation (DQC1) with postselection is equivalent in computational power to  unbounded-error space-bounded probabilistic computation,
and  the computational supremacy of the bounded-error space-bounded DQC1 is interpreted in complexity-theoretic terms.
\end{abstract}

\section{Introduction}
\label{sec:intro}
\subsection{Background}
Space-bounded computation is one of the most fundamental topics in complexity theory that have been studied in the classical and quantum settings, since it reflects common practical situations where available memory space is much less than input size. Watrous~\cite{Wat01JCSS,Wat03CC} initiated the study of space-bounded quantum computation based on quantum Turing machines and proved that, in the unbounded-error setting,  space-bounded quantum computation is equivalent in computational power to space-bounded probabilistic computation: $\PrQSPACE(s)=\PrSPACE(s)$ for any space-constructible $s$ with $s(n)\in \Omega(\log n)$.
This and the classical results~\cite{BorCooPip83InfoComp,Jun85ICALP} imply that unbounded-error space-bounded quantum computations can be simulated by deterministic computations with the squared amount of the space used by the former. 

There are some subtleties (see~\cite{MelWat12ToC} for the details) in considering space-bounded quantum computation. The most relevant one is whether we allow intermediate measurements, that is, the measurements made during computation,
which are allowed in, e.g.,~\cite{Wat01JCSS,Wat03CC,TaS13STOC}.
In the case of polynomial-time quantum computation, it is well-known that all intermediate measurements can be deferred to the end of computation 
by coherently copying the state of the qubits to be measured to ancilla qubits, and keeping their contents unchanged through the computation. This may require a polynomial number of ancilla qubits to store the copies, since there may exist a polynomial number of intermediate measurements in the original computation. This is acceptable in the polynomial-time quantum computation. However, this method is not applicable in general 
in the case of space-bounded quantum computation. For instance, if we consider a logarithmic-space quantum computation that runs in polynomial time, the above transformation may require a polynomial number of ancilla qubits, much more than the available space. Thus, it is a fundamental question in space-bounded quantum computation whether it is possible to space-efficiently eliminate intermediate measurements. Recently, Fefferman and Remscrim~\cite{FefRem21STOC} and Girish, Raz and Zhan~\cite{GirRazZha21ICALP} independently provided the breakthrough results that it is possible in the bounded error setting.

Postselection is a fictitious function that projects a quantum state on a single qubit to the prespecified state (say,  $\ket{1}$) with certainty, as far as the former state has a non-zero overlap with the latter.
Since Aaronson~\cite{Aar05RSPA} introduced postselection to the quantum complexity theory field, it has turned out to be very effective concept in the field although it is unrealistic. In particular, Aaronson succeeded in characterizing a classical complexity class in a quantum way by introducing postselection: 
$\PP=\PostBQP$~\cite{Aar05RSPA}. This characterization gives one-line proofs of the classical results~\cite{BeiReiSpi95JCSS,ForRei96InfoComp}, for which only involved classical proofs had been known, and has been a foundation for establishing the quantum computational supremacy of subuniversal quantum computation models (e.g., \cite{BreJozShe10PRS,AarArk13TOC,MorFujFit14PRL}) under complexity-theoretic assumptions. Another example of characterizing classical complexity classes with postselection is that  $\PSPACE$ is equal to $\PostQMA$~\cite{MorNis17QIC}, the class of languages that can be recognized by quantum Merlin-Arthur proof systems with polynomial-time quantum verifier with the ability of postselection. Along this line of work,  Le Gall et al.~\cite{LeGNisYak21TQC} recently considered logarithmic-space quantum computation with postselection in the bounded-error setting and proved  that  its associated complexity class $\PostBQL$ is equivalent to $\PL$, the class of languages that can be recognized with unbounded error by logarithmic-space probabilistic computation. This beautiful result can be regarded as the equivalent of $\PP=\PostBQP$ in the logarithmic-space quantum computation. 
Their model allows intermediate postselections as well as intermediate measurements, which play a key role for space-efficiency since the qubits on which  intermediate postselections or measurements are made can be reused as initialized ancilla qubits for subsequent computation. Thus, a straightforward question is whether it is possible to space-efficiently eliminate intermediate postselections and measurements. Our main result answers this question affirmatively.

\subsection{Our Contribution}
We consider the space-bounded quantum computation that allows postselections and measurements only at the end of computation, which we call space-bounded \emph{unitary} quantum computation with postselection.
Our result informally says that such quantum computation is equivalent in computational power to the space-bounded quantum computation that allows intermediate postselections and measurements.
More concretely, for a space-constructible function $s$ with $s(n)\in \Omega(\log n)$, let $\PostBQSPACE(s)$ be the class of languages that can be recognized with bounded-error by quantum computation with (intermediate) postselections and measurements that uses $O(s)$ qubits
and runs in $2^{O(s)}$ time, and let 
$\PostBQuSPACE(s)$ be the unitary version of  $\PostBQSPACE(s)$. 
By the definition, it holds $\PostBQuSPACE(s)\subseteq \PostBQSPACE(s)$.
We show the converse is also true.
\begin{theorem}\label{th:main}
For any space-constructible function $s$ with $s(n)\in \Omega(\log n)$, it holds that 
\[ \PostBQuSPACE(s)=\PostBQSPACE(s)=\PrSPACE(s). \]
\end{theorem}
This strengthens the result of $\PostBQSPACE(s)=\PrSPACE(s)$, which can be derived straightforwardly from the proof of $\PostBQL=\PL$ in~\cite{LeGNisYak21TQC,Nis21Per}.
A special case of \autoref{th:main} with $s=\log n$ is the following corollary, where
we define  $\PostBQL\equiv \PostBQSPACE(\log)$ and  $\PostBQuL\equiv \PostBQuSPACE(\log)$. 
\begin{corollary}
 $\PostBQuL=\PostBQL=\PL$.
\end{corollary}
\autoref{th:main} holds even when the completeness and soundness errors are
$2^{-2^{O(s)}}$ (see \autoref{th:PostQuSPACE} for a more precise statement).
This justifies defining the bounded-error class $\PostBQuSPACE$ (note that it is non-trivial 
to reduce errors in the space-bounded unitary computation).

As an application, we characterize the power of space-bounded computation with postselection on a quantum model that is inherently unitary.
The deterministic
quantum computation with one quantum bit (DQC1)\cite{KniLaf98PRL}, often mentioned as the one-clean-qubit model, is one of well-studied quantum computation models 
with limited computational resources (e.g.,~\cite{AmbSchVaz06JACM,ShoJor08QIC,MorFujFit14PRL,FujKobMorNisTamTan16ICALP,MorFujNis17PRA,FujKobMorNisTamTan18PRL}).
This model was originally motivated by nuclear magnetic resonance (NMR) quantum information processing, where it is difficult to initialize qubits to a  pure state.
In the DQC1 model, the initial state is thus the completely mixed state except for a single qubit,  i.e., $\ket{0}\bra{0}\otimes (\Identity/2)^{\otimes m}$,  if the total number of qubits is $m+1$. This model is inherently unitary, since if intermediate measurements or postselections were allowed, the completely mixed state could be projected to the all-zero state $\ket{0^m}$ and thus the DQC1 model would become the ordinary quantum computation model, which is supplied with the all-zero state as the initial state.

Although $\DQC1$ is considered very weak, 
a polynomial-size $\DQC1$ circuit followed by postselection is surprisingly powerful:
The corresponding bounded-error class
$\PostBDQC$
 is equal to $\PostBQP\ (=\PP)$~\cite{MorFujFit14PRL,FujKobMorNisTamTan18PRL}. 
Although this class is unrealistic, it plays an essential role in giving a strong evidence of
 computational supremacy of the $\DQC1$ computation over classical computation:
If any polynomial-size $\DQC1$ circuit is classically simulatable in polynomial time,
 then the polynomial hierarchy ($\PH$) collapses~\cite{MorFujFit14PRL,FujKobMorNisTamTan18PRL}.

Let us consider the space-bounded version of $\PostBDQC$.
For a space-constructible function $s$ with $s(n)\in \Omega(\log n)$, let $\PostBDQCSPACE(s)$ be the class of languages that can be recognized with bounded-error by DQC1 computation with postselection that uses $O(s)$ qubits
and runs in $2^{O(s)}$ time,
where all postselections and measurements are made at the the end of computation.
\begin{theorem}\label{th:DQC}
For any space-constructible function $s$ with $s(n)\in \Omega(\log n)$, it holds that
$$\PostBDQCSPACE(s)=\PostBQuSPACE(s)=\PrSPACE(s).$$
In particular, $\PostBDQCL=\PostBQuL=\PL$.
\end{theorem}
This result relates quantum computational supremacy of space-bounded  $\DQC1$ computation 
with complexity theory as in the time-bounded case.
Namely, if any $s$-space  $\DQC1$ computation can be classically simulated with space bound $s$,
then it must hold that $\PrSPACE(s)\subseteq \PostBSPACE(s)$ by \autoref{th:DQC},
where $\PostBSPACE(s)$ is the classical counterpart of $\PostBQuSPACE(s)$. 
This relation is the space-bounded equivalent of $\PP\subseteq \PostBPP$. Note that $\PP\subseteq \PostBPP$ leads to the collapse of $\PH$~\cite{BreJozShe10PRS},  since $\PostBPP$ is in the third level of $\PH$.
However, it is open whether $\PrSPACE(s)\subseteq \PostBSPACE(s)$ implies implausible consequences.

\subsection{Technical Outline}
Since space-bounded quantum computation with postselection can trivially simulate the unitary counterpart by the definition,
our main technical contribution is to show that the unitary counterpart can simulate unbounded-error probabilistic computation. 
Our starting point is the simulation of unbounded-error probabilistic computations
by space-bounded quantum computations with postselection \cite{LeGNisYak21TQC}.
The simulation~\cite{LeGNisYak21TQC} consists of two components: the first one, $Q_x$, simulates the unbounded-error probabilistic computation
with acceptance probability $p_a$ on input $x$ to output the state (up to a normalizing factor)
$\ket{\Psi_\delta}=\left(1/2+p_a\right)\left|0\right\rangle+\delta\left(1/2-p_a\right)\left|1\right\rangle$ for a given positive parameter $\delta$; the second one
decides whether $p_a>1/2$ or $p_a<1/2$ with bounded error by repeatedly running $Q_x$ to prepare the states $\ket{\Psi_\delta}$ for various values of $\delta$ 
and measuring them in the basis $\{\ket{+},\ket{-}\}$ (based on a modification of the idea~\cite{Aar05RSPA}).
We eliminate the intermediate postselection in $Q_x$ space-efficiently by,
every time postselection is made in $Q_x$,
incrementing a counter coherently if the postselection qubit is in the non-postselecting state. It is not difficult to see that this works since $Q_x$ includes only intermediate postselections (and does not include intermediate measurements). This gives a unitary version $V_x$ of $Q_x$. 
However, the second component includes both intermediate postselections and measurements, and needs run $Q_x$ sequentially (because running $Q_x$ in parallel is not space-efficient). 
We then construct two subroutines $U_+$ and $U_-$, which accumulate the amplitudes of $\ket{+}$ and $\ket{-}$, respectively, in the states obtained by repeatedly running $V_x$ for various values of $\delta$, as the amplitude of the all-zero state.
Finally, we run $U_+$ and $U_-$ in orthogonal spaces, respectively, followed by postselecting the all-zero state on the qubits that $U_+$ and $U_-$ act on.
The resulting state is significantly supported by one of the spaces, which determines whether $p_a>1/2$ or $p_a<1/2$ with high probability.

\subsection{Organization}
Sec.~2 introduces definitions, basic claims and known theorems.
Sec.~3 provides the formal statement of our main result and proves it by using a lemma, which is proved in Sec.~4.
Sec.~5 provides an application of the result.

\section{Prelimiaries}
Let $\Natural$, $\Integer$, and $\Real$ be the sets of natural numbers, integers, and real numbers, respectively. For $m\in \Natural$, let $[m]$ be the set of $\{1,\dots,m\}$. Assume that $\Sigma$ is the set $\{0,1\}$.  A \emph{promise problem} $L=(L_Y,L_N)$ is a pair of disjoint subsets of $\Sigma^\ast$. In the special case of promise problems such that $L_Y\cup L_N=\Sigma^\ast$, we say that $L_Y$ is a \emph{language}. 
\paragraph*{Classical Space-Bounded Computation}
We say that a function $s:\mathbb{N}\rightarrow\mathbb{N}$ is \emph{space-constructible} if there exists a deterministic Turing machine (DTM) that compute $s\left(\left|x\right|\right)$ in space $O\left(s\left(\left|x\right|\right)\right)$ on input $x$.
Suppose that $s$ is a space constructible function.
Then, we say that a function $f:\mathbb{N}\rightarrow\mathbb{R}$ is \emph{$s$-space computable} if there exists a DTM that computes $f\left(\left|x\right|\right)$ in space $O\left(s\left(\left|x\right|\right)\right)$ on input $x$.
For $s\in\Omega\left(\log{n}\right)$, $\PrSPACE(s)$ is the class of promise problems $L$ such that there exists a probabilistic Turing machine (PTM) $M$ running with space $O\left(s\right)$ that satisfies the following: 
For every input $x\in L_Y,$ the probability that $M$ accepts $x$ is greater than 1/2, and for every input $x\in L_N$, the probability that $M$ accepts $x$ is at most 1/2. We can replace the condition in the case of $x\in L_N$ with 
``the probability that $M$ accepts $x$ is less than 1/2'' without changing the class $\PrSPACE\left(s\right)$. In this paper, we adopt the latter definition. It is known that the class $\PrSPACE(s)$ does not change even if we impose the time bound $2^{O\left(s\right)}$ on the corresponding PTM with space bound $s$~\cite{Jun85ICALP} (see also~\cite{Sak96CCC}).
We define $\PL\equiv\PrSPACE\left(\log\right)$.

\paragraph*{Quantum Circuits}
We below introduce just notations and terminologies relevant to this paper.
For the basics of quantum computing, see standard textbooks (e.g., \cite{NieChu00Book,KitSheVya02Book}). Let $\ket{+}=(\ket{0}+\ket{1})/\sqrt{2}$ and $\ket{-}=(\ket{0}-\ket{1})/\sqrt{2}$.

A \emph{quantum gate} implements a unitary operator. We say that a quantum gate is \emph{elementary} if it acts on a constant number of qubits. We may use a quantum gate and its unitary operator that the gate implements interchangeably. Examples of elementary gates are 
$\Hadamard\equiv \ket{+}\bra{0}+\ket{-}\bra{1}$, 
$\PauliX\equiv\ket{0}\bra{1}+\ket{1}\bra{0}$,
$\magicT\equiv \ket{0}\bra{0}+e^{i\pi/4}\ket{1}\bra{1}$, and
$\CNOT\equiv \ket{0}\bra{0}\otimes \Identity +\ket{1}\bra{1}\otimes \PauliX$.
For unitary gate $g$,
$\land_k\oset{g}$ denotes the unitary gate acting on $k+1$ qubits such that it applies $g$ to the last qubit if the contents of the first $k$ qubits are all 1, and it applies the identity otherwise. We may simply say that $\land_k(g)$ is a \emph{$k$-qubit-controlled} $g$. In the case of $k=1$, we may use $\land(g)$ to denote $\land_1(g).$ For instance, $ \land_2(\PauliX)$ is the Toffoli gate and $\land(\PauliX)$ is the $\CNOT$ gate. 
We also define $\lor_k\oset{g}$ as the unitary gate acting on $k+1$ qubits such that it applies $g$ to the last qubit if the contents of the first $k$ qubits are not all-zero, and it applies the identity otherwise. 
Gate set ${G}$ is defined as $G_1\cup G_2\cup G_3$ for $G_1\equiv \set{\Hadamard, \magicT, \CNOT}$,
the set $G_2$ of a constant number of additional elementary gates used for block encoding in~\cite{LeGNisYak21TQC}, and 
the set $G_3$ of $\land (g)$ for all  $g\in G_1\cup G_2$. This choice of gates is not essential for our results when completeness and soundness errors are allowed to be $2^{-O(s)}$ for space bound $s$, because of space-efficient version of Solovay-Kitaev theorem (see~\cite[Theorem 4.3]{MelWat12ToC}), which says that it is possible to $\epsilon$-approximate every unitary gate with a sequence of gates in any fixed gate set that is finite and universal in $O(\mathrm{polylog}\left(1/\epsilon\right))$ deterministic time and $O(\log{\left(1/\epsilon\right)})$ deterministic space.

For simple descriptions, we will also use $k$-qubit-controlled gates $\land_k\left(g\right)$ for $g\in G_1 \cup G_2$ and $k\ge 2$ in the following section, since $\land_k\left(g\right)$ can be implemented with gate set $G$ together with $O(1)$ reusable ancilla qubits with negligible gate overhead:
\begin{claim}\label{cl:Cont(g)}
For every gate $g$, $\land_k(g)$ can be implemented with 
$\land(g)$, and $O(k)$  $\CNOT$ and $\magicT$ gates
together with $O(1)$ ancilla qubits initialized to $\ket{0}$. 
Similarly, $\lor_k(g)$ can be implemented with 
$\land(g)$, and $O(k)$ $\CNOT$, $\magicT$  and $\PauliX$ gates
together with $O(1)$ ancilla qubits initialized to $\ket{0}$.
Moreover, the states of the ancilla qubits return to $\ket{0}$ after applying $\land_k(g)$ ($\lor_k(g)$).
\end{claim}
\begin{proof}
One can replace $ \land_k(g)$ with a single $\land\left(g\right)$ sandwiched by two $ \land_k(\PauliX)$ with a single ancilla qubit initialized to $\left|0\right\rangle$ as their common target qubit, whose state returns to  $\left|0\right\rangle$ after $\land_k\left(g\right) $. One can further replace $\land_k(\PauliX)$ with $O(k)$ $\CNOT$ and T gates  with the help of two uninitialized ancilla qubits, whose states return to the initial state~\cite{BarBenCleDiVMarShoSleSmoWei95PRA}. 
The claim for $\lor_k(g)$ follows from de Morgan's law.
\end{proof}
Since the overhead of $O(k)$ gates can be ignored in our setting of $O(s)$ space and $2^{O\left(s\right)}$ time computation, 
we can effectively use $\land_k(g)$ freely.

Next, we define a special gate that will be used many time in this paper. 
Let $\INCREMENT_{2^n}$  be the unitary gate acting on $n$ qubits that transforms $\INCREMENT_{2^n}\colon \ket{j}\mapsto \ket{(j+1)\mod 2^n}$ for all $j\in \{0,\dots,2^n-1\}$. 
Intuitively, $\INCREMENT_{2^n}$ increments a counter over $\Integer_{2^n}$.
\begin{claim}\label{cl:ContINC}
$\land_k(\INCREMENT_{2^n})$ can be implemented with  $O(k+n^2)$ $\CNOT$ and $\magicT$  gates with the help of $O(1)$ ancilla qubits initialized to $\ket{0}$.
Similarly, $\lor_k(\INCREMENT_{2^n})$ can be implemented with  $O(k+n^2)$ $\CNOT$, $\magicT$, and $\PauliX$ gates with the help of $O(1)$ ancilla qubits initialized to $\ket{0}$.
The states of the ancilla qubits return to $\ket{0}$ after applying $\land_k(\INCREMENT_{2^n})\ (\lor_k(\INCREMENT_{2^n}))$ .

\end{claim}
\begin{proof}
By \autoref{cl:Cont(g)}, $\land_k(\INCREMENT_{2^n})$ can be implemented with $\land(\INCREMENT_{2^n})$ and $O(k)$ $\CNOT$ and $\magicT$ gates
together with $O(1)$ ancilla qubits initialized to $\ket{0}$. 
The gate $\INCREMENT_{2^n}$ can be implemented with  a single $\land_{j-1} (\PauliX)$ gate for each $j\in [n-1]$~\cite[Sec.~3.3]{FujKobMorNisTamTan16arxive} .
Thus, $\land(\INCREMENT_{2^n})$ can be implemented with  a single $\land_{j} (\PauliX)$ gate for each $j\in [n-1]$.
Since each $\land_{j} (\PauliX)$ can be implemented with $O(j)$  $\CNOT$ and $\magicT$ gates with the help of two uninitialized ancilla qubits~\cite{BarBenCleDiVMarShoSleSmoWei95PRA},  
$\land(\INCREMENT_{2^n})$ can be implemented with $O(n^2)$  $\CNOT$ and $\magicT$ gates with those ancilla qubits.
Therefore, $\land_k(\INCREMENT_{2^n})$ can be implemented with $O(k+n^2)$  $\CNOT$ and $\magicT$ gates together with $O(1)$ ancilla qubits.
The claim for $\lor_k(\INCREMENT_{2^n})$ follows from de Morgan's law.
 \end{proof}
Since the overhead of $O(k+n^2)$ gates with $n,k\in O(s)$ can be ignored in our setting of $O(s)$ space and $2^{O\left(s\right)}$ time computation, 
we can effectively use $\land_k(\INCREMENT_{2^n})$ freely. 

A \emph{quantum circuit} consists of quantum gates in a fixed universal set of a constant number of unitary gates, and (intermediate) measurements. 
A \emph{quantum circuit with postselection} is a quantum circuit with the ability of postselection even at intermediate points in the circuit.
Here, the \emph{postselection}~\cite{Aar05RSPA} is a fictitious function that projects a quantum state on a single qubit to the state $ \left|1\right\rangle$ with certainty, as far as there exists a non-zero overlap between the state and 
$\left|1\right\rangle$. For instance, if we make postselection on the first qubit of quantum state $\alpha\left|0\right\rangle\left|\psi_0\right\rangle+\beta\left|1\right\rangle\left|\psi_1\right\rangle$ with $\beta\neq0$, resulting state is $\left|1\right\rangle\left|\psi_1\right\rangle$.  Since the qubit on which postselection has been made is in the state $\left|1\right\rangle$ by the definition, we can reuse them as initialized qubits for subsequent computation. This can greatly save the space (i.e., the number of ancilla qubits) as in the case of intermediate measurements. We say that $\ket{1}$ is the postselecting state, and the state orthogonal to the postselecting state,  $\ket{0}$,  is the non-postselecting state.  To simplify descriptions, we may say ``postselect $\ket{\phi}$'' to mean that we first apply a single-qubit unitary $U$\footnote{Of course, $U$ must be efficiently implementable with gates in the gate set that we assume.} such that $\ket{1}=U\ket{\phi}$ and then postselect $\ket{1}$. We may also say ``postselect $ \left|\phi_1\right\rangle\otimes\cdots\otimes \left|\phi_m\right\rangle$,'' where $ \left|\phi_i\right\rangle$ is a single-qubit pure state for every $i$, to mean postselecting $\ket{\phi_i}$ on the $i$th qubit for each $i=1,\ldots,\ m$. For instance, we may say ``postselecting the all-zero state'' (i.e., postselecting $\left|0^m\right\rangle$).
A \emph{unitary quantum circuit} consists of only (unitary) quantum gates and does not include any measurement or postselection. To perform computational tasks, a unitary quantum circuit will be followed by measurements (and postselections).

We say that, for a promise problem $L$, a family of quantum circuits $ \{Q_x:x\in L\}$ with postselections (a family of unitary quantum circuits $\{U_x:x\in L\}$) is \emph{$s$-space uniform} if there exists a DTM that, on input $x\in L$,  outputs a description of $Q_x$ ($U_x$, respectively) with the use of space $O(s(\left|x\right|))$ (and hence in $2^{O\left(s\left(\left|x\right|\right)\right)}$ time).

Let $s\colon \Natural\to \Natural$ be a space-constructible function with $s(n)=\Omega(\log n)$.
Assume that functions $c,d\colon \Natural\to [0,1]$ are $s$-space computable, and that $c(n)>d(n)$ for sufficiently large $n\in \Natural$.
\begin{definition}[$\PostQSPACE$]\label{def:PostQSPACE}
Let $\PostQSPACE\left(s\right)\left[c,d\right]$ be the class of promise problems $L=(L_Y,L_N)$ for which there exists an $s$-space uniform family of quantum circuits with postselection, $\{Q_x:x\in L\}$,  that act on $m=O(s\left(\left|x\right|\right))$ qubits and consist of $2^{O\left(s\left(\left|x\right|\right)\right)}$ elementary gates such that,
when applying $Q_x$ to $\ket{0}^{\otimes m}$,
(1) the probability $p_{post}$ of measuring $\left|1\right\rangle$ on every postselection qubit is strictly positive;
(2) for $x\in L_Y$, conditioned on all the postselection qubits being $\left|1\right\rangle$, the probability that $Q_x$ accepts is at least $c\left(\left|x\right|\right)$;
(3) for $x\in L_N$, conditioned on all the postselection qubits being $\left|1\right\rangle$, the probability that $Q_x$ accepts is at most $d\left(\left|x\right|\right)$.
\end{definition}

\begin{definition}[$\PostQuSPACE$]\label{def:PostQuSPACE}
Let $\PostQuSPACE(s)[c.d]$ be the class of promise problems $L=(L_Y,L_N)$ for which there exists an $s$-space uniform family of unitary quantum circuits, $\set{U_x\colon x\in L}$, that act on $m\in O(s(\abs{x}))$
qubits and consist of $2^{O(s(\abs{x}))}$ elementary gates, followed by postselection on the first qubit and measurement on the output qubit  (say, the second qubit) in the computational basis, such that, when applying $U_x$ to $\ket{0}^{\otimes m}$,
(1) the probability $p_{post}$ of measuring $\ket{1}$  on the first qubits is strictly positive;
(2) for $x\in L_Y$, conditioned on the first qubit being $\ket{1}$, the probability that $U_x$ accepts  is at least $c(\abs{x})$;
(3) for $x\in L_N$, conditioned on the first qubit being $\ket{1}$, the probability that $U_x$ accepts  is at most $d(\abs{x})$.
\end{definition}
\autoref{def:PostQuSPACE} assumes that postselection is made only on a single qubit. This is general enough since, if there are $k$ postselection qubits, then we can aggregate them into a single postselection qubit by using  $\land_k(\PauliX)$ and $O(1)$ ancilla qubits.  Obviously, this aggregation does not change $p_{post}$ and the acceptance probability.

Define $\PostBQSPACE(s)\equiv \PostQSPACE(s)[2/3,1/3]$ and $\PostBQL\equiv\PostBQSPACE(\log)$. Similarly, define $\PostBQuSPACE(s)\equiv\PostQuSPACE(s)[2/3,1/3]$ and 
$\PostBQuL\equiv\PostBQuSPACE (\log)$. 
Le Gall, Nishimura and Yakaryilmaz~\cite{LeGNisYak21TQC} proved  the following.
\begin{theorem}[\cite{LeGNisYak21TQC}]
$\PostBQL=\PL.$
\end{theorem}
Moreover, it is straightforward to extend the result to the general space bound $s\in \Omega(\log n)$~\cite{Nis21Per}:
\begin{theorem}[\cite{LeGNisYak21TQC}]\label{th:PostQSPACE}
For any space-constructible function $s\colon \Natural \to \Natural$ with $s(n)\in \Omega(\log n)$, it holds that 
$\PrSPACE(s)=\PostQSPACE(s)[1-2^{-2^{O(s)}},2^{-2^{O(s)}} ]=\PostBQSPACE(s).$
\end{theorem}
In \cite{LeGNisYak21TQC}, $\PostBQL$ is defined based on the space-bounded quantum Turing machine (QTM), following the definition provided in~\cite{Wat03CC}. 
It is not difficult to see that their proof works for our circuit-based definition. Consequently, the QTM-based definition and the circuit-based definition are equivalent in computational power.
\section{Main Results}

\autoref{th:PostQuSPACE}  provides a formal statement of our main result, which shows that 
the class of promise problems that can be solved with an $s$-space uniform family of \emph{unitary} quantum circuits with postselection
by using $O(s)$ space and $2^{O(s)}$ time in the \emph{bounded-error} setting is equal to the class of promise problems by \emph{unbounded-error}
probabilistic computation with space bound $s$.

\begin{theorem}\label{th:PostQuSPACE}
For any space-constructible function $s:\mathbb{N}\rightarrow\mathbb{N}$ with $s\left(n\right)=\Omega\left(\log{n}\right)$, 
\[
\PrSPACE\left(s\right)=\PostQuSPACE\left(s\right)\left[1-2^{-2^{O\left(s\right)}},2^{-2^{O\left(s\right)}}\right]=\PostBQuSPACE\left(s\right).
\]
In particular, $\PostBQuL=\PL$.
\end{theorem}
Theorems~\ref{th:PostQSPACE} and \ref{th:PostQuSPACE} imply that
 intermediate postselections and measurements add no extra computational power, as stated formally in
the following corollary.
\begin{corollary}[Restatement of \autoref{th:main}]\label{cr:PostBQSPACE=PostBQuSPACE}
For any space-constructible function $s:\mathbb{N}\rightarrow\mathbb{N}$ with $s\left(n\right)=\Omega\left(\log{n}\right)$, 
\[ \PostBQuSPACE\left(s\right)=\PostBQSPACE\left(s\right)=\PrSPACE\left(s\right).\]
In particular, $\PostBQL=\PostBQuL$.
\end{corollary}
\begin{proof}[Proof of \autoref{th:PostQuSPACE}]
By the definition, we have 
$\PostQuSPACE\left(s\right)\left[1-2^{-2^{O\left(s\right)}},2^{-2^{O\left(s\right)}}\right]\subseteq\PostBQuSPACE\left(s\right)$.
Then, the theorem follows from Lemmas~\ref{lm:upperbound} and \ref{lm:lowerbound}, stated as follows. 
\end{proof}
\begin{lemma}
\label{lm:upperbound}
For any space-constructible function $s:\mathbb{N}\rightarrow\mathbb{N}$ with $s\left(n\right)=\Omega\left(\log{n}\right)$, 
\[ \PostBQuSPACE\left(s\right)\subseteq\PrSPACE\left(s\right).\]
\end{lemma}
\begin{proof}
By the definition, we have $\PostBQuSPACE\left(s\right)\subseteq\PostBQSPACE\left(s\right)$. Since $\PostBQSPACE\left(s\right)=\PrSPACE\left(s\right)$ by \autoref{th:PostQSPACE}, the lemma follows.
\end{proof}
\begin{lemma}\label{lm:lowerbound}
For any space-constructible function $s:\mathbb{N}\rightarrow\mathbb{N}$ with $s\left(n\right)=\Omega\left(\log{n}\right)$, 
\[ 
\PrSPACE(s)\subseteq
\PostQuSPACE(s)\left[1-2^{-2^{O\left(s\right)}},2^{-2^{O(s)}}\right].
\]
\end{lemma}
The proof is provided in the following section.

\section{Proof of~\autoref{lm:lowerbound}}
To prove~\autoref{lm:lowerbound},  we use the fact proved in~\cite{LeGNisYak21TQC}.
\begin{lemma}[\cite{LeGNisYak21TQC}]\label{lm:SimbyPost}
Suppose that, for any input $x$, a PTM with space bound $s=s\left(\left|x\right|\right)$ accepts with probability $p_a=p_a(x)$ and rejects with probability $1-p_a$ after running in a prespecified time $T(\left|x\right|)\in2^{O\left(s\right)}$. There exists an s-space uniform family of quantum circuits $Q_x$ on $m+\left\lceil\log_2{(T+1)}\right\rceil$ qubits 
for $m\in O(s)$ that consist of $2^{O\left(s\right)}$ elementary gates in the gate set $G$ and intermediate postselections such that, for every $k\in\{0,\ldots,\ T\}$, it holds that
\[
Q_x\ket{0^{m}}\ket{k}
=\frac{\ket{\Psi_k}\ket{0^{m-1}}\ket{k}}{\norm{\ket{\Psi_k}\ket{0^{m-1}}\ket{k}}},
\]
where $\ket{\Psi_k}\equiv\left(1/2+p_a\right)\left|0\right\rangle+2^{T-k}\left(1/2-p_a\right)\left|1\right\rangle$.
\end{lemma}

Let $L$ be a promise problem in $\PrSPACE(s)$. Then, there exists a PTM $M$ that recognizes $L$ with unbounded error in space $O(s\left(|x|\right))$ on input $x$. 
By the result by Jung~\cite{Jun85ICALP} (see also~\cite{Sak96CCC}),
we assume without loss of generality that $M$ runs in time $T(\abs{x})\in2^{O\left(s(\abs{x}))\right)}$.
Since the computation path is split into two paths in each step with equal probability, the accepting probability $p_a$ is of the form $a/2^T$ for some integer $a\in\{0,\cdots,2^T\}\setminus2^{T-1}$ (assuming $p_a\neq 1/2$ without loss of generality).  There exist an $s$-space uniform family of quantum circuits $Q_x$ 
defined in \autoref{lm:SimbyPost}.
Note that $Q_x$ makes intermediate postselections and thus $Q_x$ is not unitary.
Since we have assumed $p_a\neq1/2$, we can decide whether $p_a$ is larger or smaller than 1/2 with unbounded error by measuring $ \left|\Psi_T\right\rangle$ in the basis $\{\left|+\right\rangle,\left|-\right\rangle\}$. To distinguish the two cases with bounded error, we need to reduce error probability.  For this, Le Gall et al.~\cite{LeGNisYak21TQC} uses  essentially the same idea as is used in~\cite{Aar05RSPA}, 
repeating the following operations for every $k$: prepare $\left|\Psi_k\right\rangle$ and measure it in the basis $\{\left|+\right\rangle,\left|-\right\rangle\}$. Since the qubits 
on which measurements or postselections have been made
can be reused by initializing them using block encoding with postselection, the space requirement is bounded by $O\left(s\right)$. Thus, intermediate postselections and measurements play a key role in space efficiency.

\subsection{Base Unitary Circuit \texorpdfstring{$V_x$}{Vx}}\label{subsec:Vx}
Our goal is to move every postselection and measurement down to the end of computation while increasing the space requirement by at most a constant factor.

This is not difficult for the $Q_x$ part. The following modification can make $Q_x$ unitary: We prepare an $N\in O(s)$ bit counter $C$ initialized to the all-zero state $\left|0^N\right\rangle$ in a quantum register $\reg{C}$. Here, we take a sufficiently large integer in $O(s)$ as $N$.
Then, every time postselection is made in $Q_x$, we instead increment the counter $C$ coherently if the postselection qubit is in the non-postselecting state, and perform the other operations (i.e., unitary gates) in the same way as in the original circuit $Q_x$. Recall that non-postselecting state is $\left|0\right\rangle$. Thus, the counter $C$ is incremented by applying the $\PauliX$ gate to the postselection qubit, applying $\land(\INCREMENT_{2^N})$ gate controlled by that qubit, and then applying the X gate to that qubit, namely, 
$\left(\PauliX\otimes \Identity \right)\left(\land\left(\INCREMENT_{2^N}\right)\right)\left(\PauliX\otimes \Identity\right)$.
Let $V_x$ denote the modified circuit. 

By the above construction and the standard analysis (e.g., \cite{Wat01JCSS}), if we measure the counter $C$ and postselect the all-zero state after applying $V_x $ to the working and counter registers initialized to the all-zero state and the register storing the argument $k$, then the output state is $\ket{\Psi_k}$ up to a normalizing factor.
More concretely, every time postselection is made in $Q_x$, the modified circuit $V_x$ moves the non-postselection state into the space associated with the counter value being non-zero, that is, the space orthogonal to the space where the postselecting state lies. By setting $N$ so that the maximum counter value $2^{N}-1$ is larger than the number of postselections in $Q_x$, it holds that, once the counter $C$ is incremented, the counter value never returns to zero. Thus, the quantum interference between the states associated with the counter-values being zero and non-zero never occurs. This implies that, if the content of counter register $\reg{C}$ is zero, the state must be projected to the postselecting state
in \emph{every} postselection points in $Q_x$, and thus
the entire register except $\reg{C}$ is in the state that is equal to $Q_x\ket{0^m}\ket{k}$.
Thus, for certain normalized states $\ket{bad_k(j)}$ and $\ell=m+O(1)\in O(s)$,
we can write 
\begin{equation}\label{eq:Vx}
V_x\ket{0^\ell}_{\reg{R}}\ket{0^N}_{\reg{C}}\ket{k}_{\reg{K}}
=\left[\gamma_k\ket{\Psi_k}_{\reg{R}1}\ket{0^{\ell-1}}_{\reg{R}2}\ket{0^N}_{\reg{C}}+\sqrt{1-\norm{\gamma_k\ket{\Psi_k}}^2}\sum_{j\geq1}\ket{bad_k(j)}_{\reg{R}}\ket{j}_{\reg{C}}
\right]
\ket{k}_{\reg{K}},
\end{equation}
where $0<\gamma_k<1$, and
\begin{equation}\label{eq:Psik}
\ket{\Psi_k}\equiv\left(1/2+p_a\right)\ket{0}+2^{T-k}\left(1/2-p_a\right)\ket{1}
=\alpha_k\ket{+}+\beta_k\ket{-},
\end{equation}
for
$\alpha_k=\braket{+}{\Psi_k}$ and $\beta_k=\braket{-}{\Psi_k}$. Here, register $\reg{K}$ consists of $\ceil{\log_2 (T+1)}$ qubits and stores the argument $k\in \{0,\dots,T\}$. The first register $\reg{R}\equiv(\reg{R}_1,\reg{R}_2)$ is the working register except registers $\reg{C}$ and $\reg{K}$, where $\reg{R}_1$ is the subregister of $\reg{R}$ corresponding to the first qubit and $\reg{R}_2$ consists of the remaining qubits.
Note that $V_x$ uses the register $\reg{C}$  in addition to registers $(\reg{R},\reg{K})$ of $O(s)$ qubits, and applies $\land (\INCREMENT_{2^N})$ instead of every intermediate postselection made in $Q_x$. 
Since $Q_x$ consists of $2^{O\left(s\right)}$ gates in $G$ and intermediate postselections on $O(s)$ qubits, and
since \autoref{cl:ContINC} implies that
$\land (\INCREMENT_{2^N})$ is implementable with $O( N^2)\ (=O(s))$ $\CNOT$ and $\magicT$ gates with $O(1)$ ancilla qubits, which are reusable for other $\land (\INCREMENT_{2^N})$,
it holds that $V_x$ consists of $2^{O(s)}$ gates in $G$ and acts on $O(s)$ qubits.

\subsection{Subroutine Unitary Circuits \texorpdfstring{$U_+$}{U+} and \texorpdfstring{$U_-$}{U-}}
In the following subsections, we will describe the entire algorithm, which includes the error-reduction step. For this, we first provide two unitary subroutines $U_+$ and $U_-$ 
in \autoref{fig:U+U-}.
They use $V_x$ and act on registers $(\reg{R},\reg{C},\reg{D},\reg{K})$, where $\reg{D}$ is an $O(s)$-qubit quantum register used as another $O\left(s\right)$-bit counter $D$. We assume without loss of generality that the two registers $\reg{C}$ and $\reg{D}$ consist of $N$ qubits for sufficiently large $N\in O\left(s\right)$.
\begin{figure}[tb]
\hrule\vspace{3mm}
\begin{center}
\textsc{Subroutine $U_+$ associated with $V_x$}
\end{center}
Repeat the following steps $T+1$ times.
\begin{enumerate}
\item Perform $V_x$ on $(\reg{R},\reg{C},\reg{K})$. 
\item If the content of $\reg{C}$ is non-zero or the first qubit in $\reg{R}$ is in state $\left|-\right\rangle$, then apply $\INCREMENT_{2^N}$ to register $\reg{D}$ to increment the counter ${D}$.
\item If the content of $\reg{D}$ is zero, then invert the Step1 on $(\reg{R},\reg{C},\reg{K})$, i.e, apply $V_x^\dag$.
\item If the content of $(\reg{R},\reg{C})$ is not all-zero, then 
apply $\INCREMENT_{2^N}$ to register $\reg{D}$ to
increment the counter ${D}$.
\item Apply $\INCREMENT_{T+1}$ to register $\reg{K}$ to
increment the content in register $\reg{K}$.
\end{enumerate}
\vspace{3mm} 
Apply $\INCREMENT^\dag_{T+1}$ $T+1$ times
to initialize register $\reg{K}$ to the all-zero state.
\vspace{3mm}\hrule
\begin{center}
\textsc{Subroutine $U_-$ associated with $V_x$}
\end{center}\vspace{3mm}
Same as $U_+$ except that Step 2 is replaced with the following operation.
\begin{enumerate}
\setcounter{enumi}{1}
\item If the content of $\reg{C}$ is non-zero or the first qubit in $\reg{R}$ is in state $\left|+\right\rangle$, then apply $\INCREMENT_{2^N}$ to register $\reg{D}$ to increment the counter ${D}$.
\end{enumerate}
\vspace{3mm}\hrule
\caption{Subroutines $U_+$ and $U_-$ associated with $V_x$.}
\label{fig:U+U-}
\end{figure}
The following lemmas tell us about the actions of $U_+$ and $U_-$.
\begin{lemma}\label{lm:U+}
For every $k\in \{0,\dots,T\}$, let $\alpha_k$ and $\gamma_k$ be the coefficients appearing in Eqs.~(\ref{eq:Vx}) and (\ref{eq:Psik}).
Then, $U_+$ given in \autoref{fig:U+U-} acts on $O(s)$ qubits, consists of $2^{O(s)}$ gates in the gate set $G$, and satisfies
\[
U_+\ket{0}_{\reg{R}}\ket{0}_{\reg{C}}\ket{0}_{\reg{D}}
\ket{0}_{\reg{K}}=\left[
\gamma
\abs{\alpha_T}^2
\cdots
\abs{\alpha_0}^2
\ket{0}_{\reg{R},\reg{C}}
\ket{0}_{\reg{D}}
+\sum_{j\geq1}
{\ket{\phi_T(j)}_{\reg{R},\reg{C}}\ket{j}_{\reg{D}}}
\right]
\ket{0}_{\reg{K}}\]
for certain unnormalized quantum states $\left|\phi_T\left(j\right)\right\rangle$ on registers $(\reg{R},\reg{C})$ for each $j\geq1$, 
where $\gamma=\abs{\gamma_T}^2\cdots\abs{\gamma_0}^2$.
Moreover, for every $r\in \Natural\cap2^{O(s)}$, 
$(U_+)^r$ acts on $O(s)$ qubits, consists of $2^{O(s)}$ gates in the gate set $G$, and satisfies
\[
(U_+)^r\ket{0}_{\reg{R}}\ket{0}_{\reg{C}}\ket{0}_{\reg{D}}
\ket{0}_{\reg{K}}=\left[
\left(
\gamma
\abs{\alpha_T}^2
\cdots
\abs{\alpha_0}^2
\right)^r
\ket{0}_{\reg{R},\reg{C}}
\ket{0}_{\reg{D}}
+\sum_{j\geq1}
{\ket{\phi^{(r)}_T(j)}_{\reg{R},\reg{C}}\ket{j}_{\reg{D}}}
\right]
\ket{0}_{\reg{K}},
\]
where $\ket{\phi_T^{(r)}(j)}$ is a certain unnormalized quantum state on registers $(\reg{R},\reg{C})$ for each $j\geq1$.  
\end{lemma}
\begin{proof}
We first give the analysis of the first repetition on registers $(\reg{R},\reg{C},\reg{D},\reg{K})$ initialized to the all-zero state. Step 1 applies $V_x$ to 
$\ket{0^\ell}_{\reg{R}}\ket{0^N}_{\reg{C}}\ket{k}_{\reg{K}}$ with $k=0$. By \autoref{eq:Vx}, the resulting state in $(\reg{R},\reg{C})$ is 
\[
\gamma_k(\alpha_k\ket{+}+\beta_k\ket{-})_{\reg{R}1}\ket{0^{\ell-1}}_{\reg{R}2}
\ket{0}_{\reg{C}}+\sqrt{1-\norm{\gamma_k\ket{\Psi_k}}^2}
\sum_{j\geq1}{
\ket{bad_k(j)}_{\reg{R}}\ket{j}_{\reg{C}}
},
\]
where we omit the register $\reg{K}$ for simplicity. Then, Step 2 appends register $\reg{D}$ and increments the counter $D$ if the content of the register $\reg{C}$ is not zero  or  the register $\reg{R}_1$ is in the state $\left|-\right\rangle.$ Thus, the resulting state is
\begin{multline*}
\mapsto \gamma_k\alpha_k\ket{+}_{\reg{R}1}\ket{0^{\ell-1}}_{\reg{R}2}\ket{0}_{\reg{C}}
\ket{0}_{\reg{D}}\\
\quad+\left(
\gamma_k\beta_k\ket{-}_{\reg{R}1}\ket{0^{\ell-1}}_{\reg{R}2}\ket{0}_{\reg{C}}
+\sqrt{1-\norm{\gamma_k\ket{\Psi_k}}^2}\sum_{j\geq1}
{
\ket{bad_k(j)}_{\reg{R}}
\ket{j}_{\reg{C}}
}
\right)
\ket{1}_{\reg{D}}\\
=\gamma_k\alpha_k\ket{+}_{\reg{R}1}\ket{0^{\ell-1}}_{\reg{R}2}
\ket{0}_{\reg{C}}\ket{0}_{\reg{D}}+
\ket{\phi}_{\reg{R},\reg{C}}\ket{1}_{\reg{D}},
\end{multline*}
where $\ket{\phi}_{\reg{R},\reg{C}}=\gamma_k\beta_k\ket{-}_{\reg{R}1}
\ket{0^{\ell-1}}_{\reg{R}2}\ket{0}_{\reg{C}}+
\sqrt{1-\norm{\gamma_k\ket{\Psi_k}}^2}\sum_{j\geq1}{
\ket{bad_k(j)}_{\reg{R}}
\ket{j}_{\reg{C}}}$. 
Step 3 then inverts the Step1 on $(\reg{R},\reg{C},\reg{K})$, i.e, applies $V_x^\dag$, if the content of $\reg{D}$ is zero. 
Since 
${\bra{+}_{\reg{R}1}
\bra{0^{\ell-1}}_{\reg{R}2}
\bra{0}_{\reg{C}}
\bra{k}_{\reg{K}}}
V_x
{\ket{0^\ell}_{\reg{R}}\ket{0^N}_{\reg{C}}\ket{k}_{\reg{K}}
}
=\gamma_k\alpha_k,
$
the resulting state is
\[
\mapsto
\gamma_k\alpha_k\left(
(\gamma_k\alpha_k)^*
\ket{0}_{\reg{R}}
\ket{0}_{\reg{C}}
+\sqrt{1-\abs{\gamma_k\alpha_k}^2}
\ket{0^\perp}_{\reg{R},\reg{C}}
\right)
\ket{0}_{\reg{D}}+\ket{\phi}_{\reg{R},\reg{C}}
\ket{1}_{\reg{D}},
\]
where $\ket{0^\perp}$ is a certain state orthogonal to the all-zero state.
Step 4  then increments the counter ${D}$ if the content of $(\reg{R},\reg{C})$ is not all-zero; we have 
\[
\mapsto \abs{\gamma_k\alpha_k}^2
\ket{0}_{\reg{R}}
\ket{0}_{\reg{C}}
\ket{0}_{\reg{D}}
+\sum_{j=1}^{2}{\ket{\phi_k(j)}_{\reg{R},\reg{C}}
\ket{j}_{\reg{D}}},
\]
for certain states $\ket{\phi_k(j)}$ for $j=1,2$.
Step 5 increments the content in register $\reg{K}$ to get the state on $\left(\reg{R},\reg{C},\reg{D},\reg{K}\right)$:
\[
\left[\left|\gamma_k\alpha_k\right|^2\left|0\right\rangle_\reg{R}\left|0\right\rangle_\reg{C}\left|0\right\rangle_\reg{D}+\sum_{j=1}^{2}{\left|\phi_k(j)\right\rangle_{\reg{R},\reg{C}}\left|j\right\rangle_\reg{D}}\right]\left|k+1\right\rangle_\reg{K}.
\]
Then, we repeat the same procedure.  A simple induction on $k$ shows that the final state after applying $U_+$ to 
$\ket{0^\ell}_{\reg{R}}\ket{0^N}_{\reg{C}}\ket{0^N}_{\reg{D}}
\ket{0}_{\reg{K}}$ is 
\[
U_+
\ket{0^\ell}_{\reg{R}}\ket{0^N}_{\reg{C}}\ket{0^N}_{\reg{D}}\ket{0}_{\reg{K}}
=\left[
\gamma\abs{\alpha_T}^2\cdots\abs{\alpha_0}^2
\ket{0}_{\reg{R},\reg{C}}\ket{0}_{\reg{D}}
+\sum_{j\geq1}{
\ket{\phi_T(j)}_{\reg{R},\reg{C}}
\ket{j}_{\reg{D}}
}
\right]
\ket{0}_{\reg{K}},\]
where $\gamma=\left|\gamma_T\right|^2\cdots\left|\gamma_0\right|^2$. Thus, if we repeat $U_+$ $r$ times, the resulting state is
\[
(U_+)^r\ket{0^\ell}_{\reg{R}}\ket{0^N}_{\reg{C}}\ket{0^N}_{\reg{D}}\ket{0}_{\reg{K}}
=\left[
\left(
\gamma\abs{\alpha_T}^2\cdots\abs{\alpha_0}^2
\right)^r
\ket{0}_{\reg{R},\reg{C}}\ket{0}_{\reg{D}}
+\sum_{j\geq1}{
\ket{\phi^{(r)}_T(j)}_{\reg{R},\reg{C}}
\ket{j}_{\reg{D}}
}
\right]
\ket{0}_{\reg{K}},
\]
for some unnormalized states $\ket{\phi^{(r)}_T(j)}$.

Next, we consider the space and gate complexities of $U_+$ (see Appendix for a rigorous analysis).
Recall that $V_x$ can be implemented with $2^{O(s)}$ gates in $G$ by using $O(1)$ reusable ancilla qubits.
One can show 
by using Claims~\ref{cl:Cont(g)} and \ref{cl:ContINC}
that every other step in $U_+$ can also be implemented  with $2^{O(s)}$ gates in $G$ by using $O(1)$ reusable ancilla qubits.
Consequently, $U_+$ can be implemented with $2^{O(s)}$ gates in $G$  and requires $O(1)$ ancilla qubits in addition to 
$(\reg{R},\reg{C},\reg{D},\reg{K})$, which is $O(s)$ qubits in total. Since the ancilla qubits are reusable, 
$(U_+)^r$ is also implementable with $r\cdot 2^{O(s)}=2^{O(s)}$ gates in $G$ and acts on $O(s)$ qubits.
%
\end{proof}
We can prove the following lemma for $U_-$ in almost the same way.
\begin{lemma}\label{lm:U-}
For every $k\in \{0,\dots,T\}$, let $\beta_k$ and $\gamma_k$ be the coefficients appearing in Eqs.~(\ref{eq:Vx}) and (\ref{eq:Psik}).
Then, $U_-$ given in \autoref{fig:U+U-} acts on $O(s)$ qubits, consists of $2^{O(s)}$ gates in the gate set $G$, and satisfies 
\[
U_-\ket{0}_{\reg{R}}\ket{0}_{\reg{C}}\ket{0}_{\reg{D}}
\ket{0}_{\reg{K}}=\left[
\gamma
\abs{\beta_T}^2
\cdots
\abs{\beta_0}^2
\ket{0}_{\reg{R},\reg{C}}
\ket{0}_{\reg{D}}
+\sum_{j\geq1}
{\ket{\psi_T(j)}_{\reg{R},\reg{C}}\ket{j}_{\reg{D}}}
\right]
\ket{0}_{\reg{K}}\]
for certain unnormalized quantum states $\left|\psi_T\left(j\right)\right\rangle$ on registers $(\reg{R},\reg{C})$ for each $j\geq1$,
where $\gamma=\abs{\gamma_T}^2\cdots\abs{\gamma_0}^2$.
Moreover, for every $r\in \Natural\cap 2^{O(s)}$, 
$(U_-)^r$ acts on $O(s)$ qubits, consists of $2^{O(s)}$ gates in the gate set $G$, and satisfies 
\[
(U_-)^r\ket{0}_{\reg{R}}\ket{0}_{\reg{C}}\ket{0}_{\reg{D}}
\ket{0}_{\reg{K}}=\left[
\left(
\gamma
\abs{\beta_T}^2
\cdots
\abs{\beta_0}^2
\right)^r
\ket{0}_{\reg{R},\reg{C}}
\ket{0}_{\reg{D}}
+\sum_{j\geq1}
{\ket{\psi^{(r)}_T(j)}_{\reg{R},\reg{C}}\ket{j}_{\reg{D}}}
\right]
\ket{0}_{\reg{K}},
\]
where $\ket{\psi_T^{(r)}(j)}$ is a certain unnormalized quantum state on registers $(\reg{R},\reg{C})$ for each $j\geq1$.  
\end{lemma}

\begin{figure}[t]
\hrule\vspace{3mm}
\begin{center}
\textsc{Unitary Quantum Circuit with Postselection for $\PrSPACE(s)$ problems}
\end{center}
Initialize registers $(\reg{W},\reg{R},\reg{C},\reg{D},\reg{K})$ to the all-zero state.
\begin{enumerate}
\item Apply the Hadamard gate $\Hadamard$ to register $\reg{W}$.
\item If the content of $\reg{W}$ is 0, then apply $\left(U^+\right)^r$ to registers $(\reg{R},\reg{C},\reg{D},\reg{K})$; otherwise, apply $(U^-)^r$ to registers $(\reg{R},\reg{C},\reg{D},\reg{K})$.
\item Postselect the all-zero state on register $\reg{D}$.
\item Measure register $\reg{W}$ in the basis $\{\ket{0},\ket{1}\}$. If the outcome 0, then accept ($p_a>1/2$); otherwise reject (i.e., $p_a<1/2$).
\end{enumerate}
\vspace{3mm}\hrule
\caption{Unitary Quantum Circuit with Postselection for $\PrSPACE(s)$ problems}
\label{fig:finalalgo}
\end{figure}
\subsection{Final Unitary Circuit}
\autoref{fig:finalalgo} shows a unitary quantum circuit with postselection acting on 
five registers $(\reg{W},\reg{R},\reg{C},\reg{D},\reg{K})$, where $\reg{W}$ is a single-qubit register.
Now, we finalize the proof of \autoref{lm:lowerbound} with this circuit.

For simplicity, we first assume that $r=1$. It is straightforward to extend the proof to the general $r$. 
It follows from Lemmas~\ref{lm:U+} and \ref{lm:U-} that after Step 2, the state in the register 
$(\reg{W},\reg{R},\reg{C},\reg{D},\reg{K})$) is
\begin{multline}
\frac{1}{\sqrt2}\ket{0}_{\reg{W}}
\left[
\gamma
\abs{\alpha_T}^2
\cdots
\abs{\alpha_0}^2
\ket{0}_{\reg{R},\reg{C}}
\ket{0}_{\reg{D}}
+\sum_{j\geq1}
{\ket{\phi_T(j)}_{\reg{R},\reg{C}}\ket{j}_{\reg{D}}}
\right]
\left|0\right\rangle_{\reg{K}}\\
+\frac{1}{\sqrt2}\ket{1}_{\reg{W}}
\left[
\gamma
\abs{\beta_T}^2
\cdots
\abs{\beta_0}^2
\ket{0}_{\reg{R},\reg{C}}
\ket{0}_{\reg{D}}
+\sum_{j\geq1}
{\ket{\psi_T(j)}_{\reg{R},\reg{C}}\ket{j}_{\reg{D}}}
\right]
\left|0\right\rangle_{\reg{K}}.
\end{multline}
Step 3  postselects the all-zero state in register $\reg{D}$. Thus, the resulting state in the register $(\reg{W},\reg{R},\reg{C})$ is
\begin{equation}\label{eq:postselectedstate}
\gamma^\prime
\left(
\abs{\alpha_T}^2
\cdots
\abs{\alpha_0}^2
\ket{0}_{\reg{W}}
+
\abs{\beta_T}^2
\cdots
\abs{\beta_0}^2
\ket{1}_{\reg{W}}
\right)
\left|0\right\rangle_{\reg{R},\reg{C}}
\end{equation}
where $\gamma^\prime$ is the normalizing factor.
Here, the probability of measuring postselecting state is 
\begin{equation*}
p_{post}=\frac{1}{2}\left(\left(
\gamma
\abs{\alpha_T}^2
\cdots
\abs{\alpha_0}^2
\right)^2 
+\left(
\gamma
\abs{\beta_T}^2
\cdots
\abs{\beta_0}^2
\right)^2
\right).
\end{equation*}
Recall that $\alpha_k=\braket{+}{\Psi_k}$ and $\beta_k=\braket{-}{\Psi_k}$, where $\ket{\Psi_k}$ is defined as \autoref{eq:Psik}.
If $p_a>1/2$,  then $\beta_k>0$ for all $k$. If $p_a<1/2$, then $\alpha_k>0$ for all $k$. Since $\gamma\neq0$, $p_{post}$ is strictly positive in both cases.

If $p_a<1/2$, then $\abs{\alpha_k}^2>\abs{\beta_k}^2\geq 0$ for all $k$, and $\left|\alpha_k\right|^2>\left(1+\delta\right)\left|\beta_k\right|^2$ for some $k$ and some constant $\delta$, say, $16/9$.  Since
$\frac{\left|\beta_T\right|^2\cdots\left|\beta_0\right|^2}{\left|\alpha_T\right|^2\cdots\left|\alpha_0\right|^2}<\frac{1}{1+\delta}=\frac{9}{25}$,
the probability that $\left|0\right\rangle_{\reg{W}}$ is measured in Step 4,
that is, the probability of obtaining the outcome 0 when measuring register $\reg{W}$ in~\autoref{eq:postselectedstate} in the basis $\{ \ket{0},\ket{1}\}$,
is
\begin{equation*}
\frac{
\abs{\alpha_T}^4
\cdots
\abs{\alpha_0}^4
}{
\abs{\alpha_T}^4
\cdots
\abs{\alpha_0}^4
+
\abs{\beta_T}^4
\cdots
\abs{\beta_0}^4
}=\frac{1}{
1+\left(
\abs{\beta_T}^4
\cdots
\abs{\beta_0}^4
\right)/\left(
\abs{\alpha_T}^4
\cdots
\abs{\alpha_0}^4
\right)
}>\frac{1}{1+(9/25)^2}=\frac{625}{706}.
\end{equation*}
If $p_a>1/2$, then $0\le \abs{\alpha_k}^2<\abs{\beta_k}^2$ for all $k$, and $(1+\delta)\abs{\alpha_k}^2<\abs{\beta_k}^2$ for some $k$ and a constant $\delta=16/9$. Thus, 
the probability that $\ket{1}_{\reg{W}}$ is measured in Step 4
is at least $\frac{625}{706}$ in the same analysis as in the case of $p_a<1/2$.

For a general $r\in 2^{O(s)}$, the state in the register $(\reg{W},\reg{R},\reg{C})$ after Step 3 is
\begin{equation*}
\gamma^{\prime\prime}
\left(
\left(
\abs{\alpha_T}^2
\cdots
\abs{\alpha_0}^2
\right)^r
\ket{0}_{\reg{W}}
+
\left(
\abs{\beta_T}^2
\cdots
\abs{\beta_0}^2
\right)^r
\ket{1}_{\reg{W}}
\right)
\left|0\right\rangle_{\reg{R},\reg{C}},
\end{equation*}
where $\gamma''$ is the normalizing factor.
For $p_a<1/2$, thus, the probability that $\ket{0}_{\reg{W}}$ is measured in Step 4 is
\begin{equation*}
\frac{
\Oset{
\abs{\alpha_T}^4
\cdots
\abs{\alpha_0}^4
}^r
}{
\Oset{
\abs{\alpha_T}^4
\cdots
\abs{\alpha_0}^4
}^r
+
\Oset{
\abs{\beta_T}^4
\cdots
\abs{\beta_0}^4
}^r
}
>\frac{1}{1+(9/25)^{2r}}>1-\frac{81^r}{625^r+81^r}>1-\frac{1}{2^r}.
\end{equation*}
Similarly, for $p_a>1/2$, the probability that $\ket{1}_{\reg{W}}$ is measured in Step 4 is at least $1-\frac{1}{2^r}$.

Finally, we consider the space and gate complexities. 
The quantum circuit in \autoref{fig:finalalgo} acts on a single-qubit register $\reg{W}$ in addition to $(\reg{R},\reg{C},\reg{D},\reg{K})$. 
In this circuit, 
every gate $g\in G$ used in $(U_+)^r$ and $(U_-)^r$ is replaced with  $\land(g)$, which can be implemented with $O(1)$ gates in $G$  with $O(1)$ resuable ancilla qubits
by \autoref{cl:Cont(g)}.
Since Step 2 is dominant, and $(U_+)^r$ and $(U_-)^r$ use $O(s)$ qubits and $2^{O(s)}$ gates in $G$ by  Lemmas~\ref{lm:U+} and \ref{lm:U-},
the entire circuit uses $O(s)$ qubits and $ 2^{O(s)}$ gates in $G$.

\section{Application to One-Clean Qubit Model}
As introduced in \autoref{sec:intro},
$\DQC1$ is a model of quantum computing such that the input state is  completely mixed except for one qubit, which is initialized to $\ket{0}$.  
\begin{definition}[$\PostDQCSPACE$]
Let $s$ be any space-constructible function with $s(n)\in \Omega(\log n)$,
and let $c,d$ be $s$-space computable functions.
$\PostDQCSPACE(s)[c,d]$ is the class of promised problems $L=(L_Y,L_N)$  for which there exists  
an $s$-space uniform family of unitary quantum circuits $\{U_x: x\in L\}$ consisting of $2^{O(s)}$ elementary gates on $m+1$ qubits for $m\in O(s)$
such that,  when applying $U_x$ to the $m+1$ qubits in state $\ketbra{0}\otimes (\Identity/2)^{\otimes m}$,
followed by postselections and measurements,
(1) the probability $p_{post}$ of measuring $\ket{1}$ on all postselection qubits is strictly positive;
(2) for $x\in L_Y$, conditioned on all the postselection qubits being $\ket{1}$, the probability that $U_x$ accepts is at least $c(|x|)$;
(3) for $x\in L_N$, conditioned on all the postselection qubits being $\ket{1}$, the probability that $U_x$ accepts is at most $d(|x|)$.
In particular, define $\PostBDQCSPACE(s)\equiv \PostDQCSPACE(s)[2/3,1/3]$, and $\PostBDQCL\equiv\PostBDQCSPACE(\log)$.
\end{definition}
In the above definition, we allow postselection to be made on multiple qubits, since it does not seem possible in general to aggregate multiple postselection qubits to a single qubit due to the lack of initialized qubits.
\autoref{th:DQC} follows from Theorems~\ref{th:PostQuSPACE} and \ref{th:DQCprecise}.

\begin{theorem} \label{th:DQCprecise}
For any space-constructible function $s$ with $s(n)\in \Omega(\log n)$
and any $s$-space computable functions $c$ and $d$ such that $c(n)>d(n)$ for sufficiently large $n\in \Natural$,
\[
\PostDQCSPACE(s)[c,d]=\PostQuSPACE(s)[c,d].
\]
In particular, $\PostBDQCSPACE(s)=\PostBQuSPACE(s)$ and $\PostBDQCL=\PostBQuL$.
\end{theorem}
\autoref{th:DQCprecise} immediately follows from Lemmas~\ref{lm:upperDQC} and \ref{lm:lowerDQC}.
\begin{lemma} \label{lm:upperDQC}
$\PostDQCSPACE[c,d]\subseteq\PostQuSPACE[c,d]$. 
\end{lemma}
\begin{proof}
For any promise problem $L\in \PostDQCSPACE[c,d]$, let $\{U_x\colon x\in L\} $ be the $s$-uniform family of unitary quantum circuits given in the definition of $\PostDQCSPACE[c,d]$. Recall that all qubits but one are dirty, i.e., initially in the completely mixed state.  This is only the difference between the two complexity classes (except for the possible number of postselection qubits).
Fortunately, the completely mixed state can be generated easily from the all-zero state as follows.
For each qubit in the completely mixed state,  first prepare the maximally entanglement state $\left(\left|00\right\rangle+\left|11\right\rangle\right)/\sqrt2$  by applying $\Hadamard\otimes \Identity$ and then applying $\CNOT$ to two qubits in the state $\ket{0}\otimes \ket{0}$, and then discarding one of the two qubits. This requires at most twice the original space. Hence, the space is still bounded by $O(s)$.  If there are multiple postselection qubits in $U_x$, then we can aggregate them into a single postselection qubit space-efficiently. 
By putting all things together, we obtain another unitary circuit $U'_x$ acting on $O(s)$ qubits, which are initialized to the all-zero state.
It is easy to see that the completeness and soundness are still $c$ and $d$, respectively.
Therefore, $L$ is in $\PostQuSPACE[c,d]$. 
\end{proof}
\begin{lemma}\label{lm:lowerDQC}
$\PostQuSPACE(s)[c,d]\subseteq \PostDQCSPACE(s)[c,d]$.
\end{lemma}
\begin{proof}
For $L\in \PostQuSPACE[c,s]$, there exists an $s$-space uniform family of
unitary quantum circuits $\{ U_x\colon x\in L\}$ followed by postselection that recognize $L$ 
with completeness $c$ and soundness $d$ on input $x$.
Let $p$ and $o$ be the postselection qubit and the output qubit (to be measured), respectively, 
of $U_x$. Suppose that $U_x$ acts on $m\in O(s)$ qubits.
Then, one can simulate $U_x$ by a $\DQC1$ circuit over $m+1$ qubits
with postselection in the standard way.
Namely, first flip the state of the unique clean qubit if the content of all the dirty qubits are zeros by using $\land_m(\PauliX)$. 
Note that $\land_m(\PauliX)$ can be implemented with $O(m)$ CNOT and $T$ gates by using $O(s)$ ancilla qubits in any unknown initial states, e.g., completely mixed states~\cite{BarBenCleDiVMarShoSleSmoWei95PRA}.
Then, apply $U_x$  to the $m$ dirty qubits,
postselect $\ket{1}$ on the clean qubit and the qubit $p$, and then measure the qubit $o$.
It is easy to see that the completeness and soundness are still $c$ and $d$, respectively.
\end{proof}


\section*{Acknowledgements}
I am grateful to Harumichi Nishimura for helpful discussions and answering my questions about his co-authoring paper~\cite{LeGNisYak21TQC}. I am also grateful to Yasuhiro Takahashi and Yuki Takeuchi for valuable comments. This work was partially supported by JSPS KAKENHI Grant Number JP20H05966 and JP22H00522.

\bibliographystyle{plain}
\bibliography{SpaceBoundPost}

\appendix
\section*{Appendix}\label{appdx}
\section{Space and Gate Complexities of \texorpdfstring{$U_+$}{U+}
in the Proof of \autoref{lm:U+}}\label{appdx:GatesInU+}
To show that $U_+$ can be implemented with $2^{O(s)}$ gates in $G$ and acts on $O(s)$ qubits,
we examine each step in \autoref{fig:U+U-}. As for the space complexity, we show that each step requires $O(1)$ \emph{reusable} ancilla qubits,
meaning that the states of the ancilla qubits return to the initial states after each step, and thus the total number of ancilla qubits is bounded by $O(1)$.
Since the implementation of $U_+$ with elementary gates acts on registers $(\reg{R},\reg{C},\reg{D},\reg{K})$  of $O(s)$ qubits except $O(1)$ ancilla qubits,
the number of qubits on which it acts is bounded by $O(s)$.
Now, let us analyze each step.

Step 1 (i.e., $V_x$) uses $2^{O(1)}$ gates in $G$ with $O(1)$ reusable ancilla qubits as we have analyzed in \autoref{subsec:Vx}.

To perform Step 2,  we prepare three ancilla qubits $q_1,q_2,q_3$ initialized to $\ket{0}$.
We first use $\lor_N(\PauliX)$ to flip the state of $q_1$  (i.e., apply an $\PauliX$ gate to $q_1$)  if
the content of $\reg{C}$ is non-zero, 
and use the circuit $W\equiv (\Hadamard\otimes \Identity)\CNOT(\Hadamard\otimes \Identity)$ 
to flip the state of $q_2$ if the first qubit in $\reg{R}$ (i.e., $\reg{R}1$) is in the state $\ket{+}$. 
Then, we use the circuit $(\Identity\otimes \Identity\otimes \PauliX) (\land_2(\PauliX))(\PauliX\otimes X\otimes \Identity)$ on qubits $q_1,q_2,q_3$ 
to take the $\OR$ of the contents of  $q_1,q_2$ and store the result to $q_3$. 
Finally, we apply $\land(\INCREMENT_{2^N})$ to qubit $q_3$ and register $\reg{D}$, where the control qubit is $q_3$. 
The three qubits $q_1,q_2,q_3$ are made to return to the initial states $\ket{0}$ by inverting the $\OR$, $W$, and $\lor_N(\PauliX)$ 
so that they can be reused for later computation. 
It is almost clear from  \autoref{cl:Cont(g)} and \autoref{cl:ContINC}
that Step 2 is implementable with $2^{O(s)}$ gates.
More concretely, by using $O(1)$ reusable ancilla qubits,
$\lor_N(\PauliX)$ can be implemented with $O(N)$  $\PauliX$, $\magicT$, and $\CNOT$ gates by \autoref{cl:Cont(g)}; 
$\OR$ can be implemented  with $O(1)$  $\PauliX$, $\magicT$, and $\CNOT$ gates;
$\land(\INCREMENT_{2^N})$ can be implemented with $O(N^2)\ (=O(s^2))$ of $\CNOT$ and $\magicT$ gates
by \autoref{cl:ContINC}.

Step 3 applies $\land_N(V_x^\dag)$ with each controlling qubit being sandwiched by $\PauliX$ gates.
Since $V_x$ is implemented with $2^{O(s)}$ gates in $G$ by using $O(1)$ reusable ancilla qubits, $\land_N(V_x^\dag)$ can be implemented with $2^{O(s)}$ gates in $G$ 
by using  $O(1)$  ancilla qubits by \autoref{cl:Cont(g)}.

Step 4 applies $\lor_{\ell+N}(\INCREMENT_{2^N})$, which can be implemented with 
$O((\ell+N)+N^2)\ (=O(s^2))$ $\PauliX$, $\magicT$, and $\CNOT$ gates and $O(1)$ reusable ancilla qubits by \autoref{cl:ContINC}.

Step 5 can be implemented with $O((\log T)^2)\ (=O(s^2))$ $\CNOT$ and $\magicT$ gates and $O(1)$ reusable ancilla qubits by \autoref{cl:ContINC}.

Since Steps 1 to 5 are repeated $T+1\in 2^{O(s)}$ times, the total number of gates for these steps is still $2^{O(s)}$. The final step of applying $\INCREMENT^\dag_{T+1}$ $T+1$ times consists of $O(T(\log T)^2)\ (=2^{O(s)})$ $\CNOT$ and $\magicT$ gates with $O(1)$ reusable ancilla qubits.
Consequently, $U_+$ can be implemented with $2^{O(s)}$ gates in $G$ and $O(1)$ reusable ancilla qubits. \hspace*{\fill}\qed

\end{document}